\documentclass{ifacconf}
\usepackage{graphicx}      
\usepackage{natbib}        
\usepackage{cuted}
\usepackage{siunitx}
\usepackage{tabularx}

\usepackage{float}
\usepackage{color}
\usepackage{amssymb}
\makeatletter
\let\theoremstyle\@undefined
\makeatother
\usepackage{amsthm}
\usepackage{amsmath}
\usepackage{bm} 

\newcommand{\DM}[1]{\textcolor{red}{#1}}

\newcommand{\Ha}{\mathcal{H}}

\renewcommand{\d}{\mathrm{d}}

\newcommand{\Dom}{\Omega}
\newcommand{\Domb}{\Gamma}
\newcommand{\inner}[2]{\left\langle #1 \right\rangle_{#2}}

\newcommand{\grad}{\mathbf{grad}} 
\newcommand{\tr}{\operatorname{tr}}

\newcommand{\D}{\mathbb{G}\text{rad}~} 
\DeclareMathOperator*{\Div}{Div}
\renewcommand{\div}{\operatorname{div}}

\newtheorem{theorem}{Theorem}[section]

\newtheorem{definition}[theorem]{Definition}
\newtheorem{proposition}[theorem]{Proposition}
\newtheorem{lemma}{Lemma}

\newcommand{\R}{\mathbb R}

\newcommand{\LR}{L^2(\Dom; \R)}

\newcommand{\Rn}{\R^n}
\newcommand{\LRn}{L^2(\Dom; \Rn)}

\newcommand{\vx}{\bm{x}}
\newcommand{\vy}{\bm{y}}
\newcommand{\ve}{\bm{e}}
\newcommand{\vf}{\bm{f}}
\newcommand{\vg}{\bm{g}}
\newcommand{\vn}{\bm{n}}

\newcommand{\vr}{\bm{r}}

\newcommand{\vG}{\bm{G}}

\newcommand{\Tens}{\R^{n\times n}_{\text{sym}}}
\newcommand{\LTens}{L^2(\Dom; \Tens)}

\newcommand{\TG}{\mathbb G}
\newcommand{\TLa}{\Lambda}

\newcommand{\TX}{\mathbb X}
\newcommand{\TY}{\mathbb Y}
\newcommand{\TE}{\mathbb E}
\newcommand{\TF}{\mathbb F}

\begin{document}
\begin{frontmatter}

\title{Conduction-Diffusion in $N$-Dimensional settings as irreversible port-Hamiltonian systems \thanksref{footnoteinfo}} 

\thanks[footnoteinfo]{This work has been achieved in the frame of the EIPHI Graduate school (contract ``ANR-17-EURE-0002") and ANID funded projects CIA250006 and FONDECYT 1231896.}
\author[First]{Luis Mora} 
\author[Second]{Yann Le Gorrec} 
\author[Third]{Hector Ramirez}
\author[Fourth]{Denis Matignon}

\address[First]{Dept. of Applied Mathematics, University of Waterloo, Waterloo, Canada}
\address[Second]{Université Marie et Louis Pasteur, SUPMICROTECH, CNRS, institut FEMTO-ST, F-25000 Besançon, France}
\address[Third]{Departamento de Electronica, Universidad Tecnica Federico Santa Maria, Valparaiso, 2390123, Chile}
\address[Fourth]{Fédération ENAC ISAE-SUPAERO ONERA, Université de Toulouse, Toulouse, 31055, France}

\begin{abstract}                
This work extends previous $1D$ irreversible port-Hamiltonian system (IPHS) formulations to boundary-controlled $ND$ distributed parameter systems describing conduction–diffusion fluid phenomena. Within a unified and thermodynamically consistent framework, we show that conduction and diffusion can be represented through a single coherent structure that preserves global energy balance and ensures a correct characterization of entropy production. The resulting formulation provides a foundation for the systematic modeling and control of complex multi-physical processes governed by coupled transport mechanisms in $N$-dimensions. In the longer term, this framework opens the door to structure-preserving numerical schemes capable of enforcing thermodynamic principles directly at the discretized level.
\end{abstract}

\begin{keyword}
Irreversible port-Hamiltonian systems, Thermodynamics, Diffusion, Convection, N-Dimensional fluids 
\end{keyword}

\end{frontmatter}
\allowdisplaybreaks

\section{Introduction}

Thermodynamic modeling of fluids plays a central role in predicting and analyzing physical and chemical behavior under variations of pressure, temperature, and composition \citep{Bird_2006}. Such modeling relies on energy-based constitutive laws, equations of state, and transport relations that characterize phase behavior, density, internal energy, entropy production, and reactive transformations \citep{de_Groot_1962}. Accurate thermodynamic representations are essential for describing a wide range of engineering processes \citep{Kjelstrup_2017}, including combustion, reactive mixing, mineral dissolution, aerothermal flows, and multi-phase transport, where both the first and second laws of thermodynamics constrain the admissible system dynamics. In recent years, the need for thermodynamically consistent models has intensified, driven by modern applications where energy fluxes, irreversible phenomena, and multi-physics couplings dominate system behavior \citep{Dubljevic_2022}. 
A promising approach to address these challenges builds on the Port-Hamiltonian system (PHS) framework \citep{Geoplex09}, originally developed as an extension of classical Hamiltonian mechanics to open systems interacting with their environment. PHS formulations encode the interconnections between energy storage and power exchange through a skew symmetric geometric structure and power conjugated port variables. These properties make PHS particularly suitable for the modular modeling of complex multi domain systems, where mechanical, thermal, and chemical phenomena are tightly coupled. Generalizations to distributed parameter settings \citep{Gorrec_2005,Rashad_Califano_2020} have introduced boundary port variables derived from the co-states, enabling the representation of open fluid flows, heat transfer mechanisms, reactive mass fluxes, and fluid-structure interactions. 
Within fluid mechanics, PHS and pseudo PHS formulations have been proposed for various models, including inviscid flows, shallow water dynamics, and incompressible Navier-Stokes systems. More recently, energy based descriptions have been extended to compressible and reactive fluids, providing coherent representations of the energy fluxes that govern thermodynamic behavior \citep{Califano_2022,MoraPoF2021}. These extensions are particularly relevant because modern computational fluid dynamics methods often require artificial corrections to approximate physical energy fluxes or struggle with nonphysical instabilities when thermodynamic consistency is not explicitly enforced. A formulation that inherently incorporates energy flux, such as the PHS framework, offers a systematic alternative that naturally aligns modeling, simulation, and control design \citep{CardosoJCF2025}.
A further development in this context is the Irreversible Port-Hamiltonian Systems (IPHS) framework, originally introduced in \citep{Ramirez_CES_2013,Ramirez_EJC_2013} and extended to 1D dimensional non-convective distributed parameter systems in \citep{RamirezCES2022,Ramirez2022Entropy}. IPHS generalizes traditional PHS to explicitly incorporate irreversible thermodynamics. Unlike classical PHS, IPHS formulations use the total energy as generating function while embedding the entropy balance directly within the state evolution, thereby guaranteeing consistency with both the first and second laws of thermodynamics. The resulting structure, a modulated skew symmetric operator, captures not only system topology but also the irreversible thermodynamic forces induced by viscosity, heat conduction, diffusion, and chemical reactions. These features make IPHS particularly appealing for passivity based and energy shaping control strategies \citep{Ortega_IEEE_CSM_2001}, including IDA-PBC approaches \citep{Ramirez_Automatica_2016}, which directly exploit the energetic structure for feedback design.
Motivated by the examples in \citep{RamirezLeGorrecTFMST2016}, and building upon the generalization of IPHS to $ND$ systems provided in \citep{MoraWC2023}, this work extends the 1D IPHS formulations of \citep{RamirezCES2022} to boundary controlled $ND$ distributed parameter systems for conduction-diffusion fluid phenomena within a unified, thermodynamically consistent framework. This extension integrates conduction, diffusion and reactive transport into a single coherent formalism that preserves energy balance and entropy production. 
 The paper is organized as follows: Irreversible port-Hamiltonian systems in 1D are recalled  \S~\ref{s-1DIpHs}. Then, the multidimensional heat equation is recast as an $ND$ IPHS in \S~\ref{s-Cond}. The complete process with the diffusion is then presented in  \S~\ref{s-CDP}, first with one species only, then with $n$ species. 

\section{1D irreversible port-Hamiltonian systems}
\label{s-1DIpHs}
Irreversible port-Hamiltonian formulations have been extended to distributed parameter systems defined on a one-dimensional spatial domain\footnote{Excluding systems with convection.} $z \in [a,b],\,a,\,b\in\mathbb{R},\,a<b$ in \citep{RamirezCES2022}. The state variables of the system are the $n+1$ \emph{extensive variables}\footnote{A variable is qualified as extensive when it characterizes the thermodynamic state of the system and its total value is given by the sum of its constituting parts.} of Thermodynamics, $\mathbf{x} \in\mathbb{R}^{n+1}$ : the first $n$ variables  $x=[q_1,\ldots,q_n]^\top \in\mathbb{R}^n$ are the $n$ extensive variable of thermodynamics associated to the involved physical domain excluding the thermal one, and the remaining component represents the entropy density $s \in \mathbb{R}$. 
The energy function is in this case defined by
\begin{equation}\label{eq:Totalenergy}
H(x,s)=\int_a^b h\left(x(z),s(z)\right) dz
\end{equation} 
where $h(x,s)$ is the energy density function. The total entropy functional is denoted by
\begin{equation}\label{eq:Totalentropy}
S(t)=\int_a^b s(z,t) dz
\end{equation}
We shall furthermore use the following notation. For any two functionals $H_1$ and $H_2$ of the type (\ref{eq:Totalenergy}) and for any matrix differential operator ${\mathcal G}$ we define the pseudo-brackets 
\begin{equation}\label{pseudo-bracket}
\begin{split}
\left\{H_1|{\mathcal G}|H_2 \right\} & =
\begin{bmatrix}
\frac{\delta H_1}{\delta x} \\ 
\frac{\delta H_1}{\delta s}
\end{bmatrix}^\top
\begin{bmatrix}
0 & {\mathcal G}\\
-{\mathcal G}^* & 0 \end{bmatrix}
\begin{bmatrix}
\frac{\delta H_2}{\delta x} \\ 
\frac{\delta H_2}{\delta s}
\end{bmatrix}, \\ 
\left\{H_1|H_2\right\} & = \frac{\delta H_1}{\delta s} \,\grad\left(\frac{\delta H_2}{\delta s}\right)
\end{split}
\end{equation} 
where ${\mathcal G}^*$ denotes the formal adjoint operator of ${\mathcal G}$.

\begin{definition}
\label{definition_IPHS}
An infinite-dimensional IPHS undergoing $m$ irreversible processes is defined by 
\begin{itemize}
\item a pair of functionals: the total energy (\ref{eq:Totalenergy}) and the total entropy (\ref{eq:Totalentropy})
\item a pair of matrices $P_0=-{P}^\top_0 \in \mathbb{R}^{n\times n}$ and $P_1={P}^\top_1 \in \mathbb{R}^{n\times n}$
\item a pair of matrices  $G_0 \in \mathbb{R}^{n\times m}$, $G_1 \in \mathbb{R}^{n\times m}$ with $m\leq n$ and the strictly positive real-valued functions $\gamma_{k,i}\left(x,s,\tfrac{\delta H}{\delta x},\tfrac{\delta H}{\delta s}\right) \: k=0,\,1;\: i\in \left\{ 1,\,...\,m \right\}$
\item a pair of real-valued functions $\gamma_{s}\left(x,s,\tfrac{\delta H}{\delta x},\tfrac{\delta H}{\delta s}\right)>0$ 
and $g_s(x)$
\end{itemize}
and the PDE
\begin{multline}\label{IPHS-BCS}
\frac{\partial }{\partial t} 
\begin{bmatrix} 
x(t,z)\\ 
s(t,z)
\end{bmatrix}=
\begin{bmatrix}
P_0 & G_0\mathbf{R_0}\\ 
-\mathbf{R_0}^\top G_0^\top &0
\end{bmatrix}
\begin{bmatrix}
\frac{\delta H}{\delta x}(t,z)\\ 
\frac{\delta H}{\delta s}(t,z) 
\end{bmatrix}+
  \\ \begin{bmatrix}P_1 \frac{\partial (.)}{\partial z}  & &\frac{\partial \left(G_1 \mathbf{R_1} .\right) }{\partial z}\\ \mathbf{R_1}^\top G_1^\top  \frac{\partial\left(.\right) }{\partial z}&& g_s \mathbf{r_s} \frac{\partial \left(.\right)}{\partial z}+\frac{\partial \left(g_s \mathbf{r_s}.\right)}{\partial z}\end{bmatrix}\begin{bmatrix}\frac{\delta H}{\delta x}(t,z)\\ \frac{\delta H}{\delta s}(t,z) \end{bmatrix}
\end{multline}
with vector-valued functions $\mathbf{R_l}\left(x,s,\tfrac{\delta H}{\delta x},\tfrac{\delta H}{\delta s}\right) \in \mathbb{R}^{m\times 1}$, $l=0,1$, defined by
\begin{equation*}
R_{0,i}=\gamma_{0,i} \left\{S|G_0(:,i)|H\right\}
\end{equation*}
\begin{equation*}
R_{1,i}=\gamma_{1,i} \left\{S|G_1(:,i)\tfrac{\partial}{\partial z}|H\right\}
\end{equation*}
and  
\begin{equation*}
r_{s}=\gamma_s\left\{S|H\right\}
\end{equation*}
where the notation $G(:,i)$ indicates the $i$-th column of the matrix $G$.
\end{definition}

\begin{definition}\label{definition_BC_IPHS}
A Boundary Controlled IPHS (BC-IPHS) is an infinite-dimensional IPHS according to Definition \ref{definition_IPHS}, augmented with the boundary port variables
\begin{align}\label{eq:IPHS_input}
v(t)  & = W_{B}
\begin{bmatrix}
e(t,b) \\
e(t,a)
\end{bmatrix},&  &y(t)  = W_{C}
\begin{bmatrix}
e(t,b) \\
e(t,a)
\end{bmatrix} 
\end{align}
as linear functions of the modified effort variable
\begin{equation}\label{boundary_variables_IPHS}
e(t,z)  = 
\begin{bmatrix}
\frac{\delta H}{\delta x}\\
\mathbf{R} \:\frac{\delta H}{\delta s}  \\
\end{bmatrix},
\end{equation}
with $\mathbf{R} =\begin{bmatrix}1 & \mathbf{R_1} & \mathbf{r_s} \end{bmatrix}^\top$ and
\begin{align*}
W_{B} &= 
\begin{bmatrix}
\frac{1}{\sqrt{2}}\left(\Xi_2 + \Xi_1P_{ep} \right)M_p & \frac{1}{\sqrt{2}} \left(\Xi_2 - \Xi_1 P_{ep} \right)M_p
\end{bmatrix},
\\
W_{C}& = 
\begin{bmatrix} 
\frac{1}{\sqrt{2}}\left(\Xi_1+\Xi_2 P_{ep}\right)M_p  & \frac{1}{\sqrt{2}}\left(\Xi_1-\Xi_2 P_{ep}\right)M_p 
\end{bmatrix},
\end{align*}   
where $M_p=\left( M^\top M\right)^{-1}M^\top$, $P_{ep}=M^\top P_e M$ and $M\in \mathbb{R}^{(n+m+2) \times k} $ is spanning the columns of $P_e \in \mathbb{R}^{n+m+2}$ of rank $k$, defined by\footnote{$0$ has to be understood as the zero matrix of appropriate dimensions.}
\begin{equation}\label{def:Pe}
P_e=\begin{bmatrix}
 P_1 & 0 &G_1 & 0\\
0 & 0 &0 & g_s\\
G_1^\top & 0 &0 & 0\\
0 & g_s  & 0  & 0\\
\end{bmatrix}
\end{equation}
and where $\Xi_1$ and $\Xi_2$ in $\mathbb{R}^{k \times k} $satisfy $\Xi_2^\top\Xi_1+\Xi_1^\top\Xi_2=0$ and $\Xi_2^\top\Xi_2+\Xi_1^\top \Xi_1=I$.\rule{0.5em}{0.5em}
\end{definition}

With this formulation at hand, the first principle of Thermodynamics can be recovered in:
\begin{lemma} [First law of Thermodynamics] 
The total energy balance is 
\begin{equation*}
\dot{H} = y(t)^\top v(t)^,,
\end{equation*}
which leads, when the input is set to zero, to $\dot{H} = 0$ in accordance with the first law of Thermodynamics.
\end{lemma}
\begin{proof}
The energy conservation is a direct consequence of the skew symmetry of the differential operator (cf. \cite{RamirezCES2022}).
\end{proof}

Similarly, this formulation encodes the second principle of Thermodynamics that is expressed in:

\begin{lemma} [Second law of Thermodynamics] 
The total entropy balance is given by  
\begin{equation*}
\dot{S} =  \int_a^b \sigma_t dz -y_{S}^\top v_{s}\,,
\end{equation*}
where $y_s$ and $v_s$ are the entropy conjugated input/output and $\sigma_t$ is the total internal entropy production. This leads, when the input is set to zero, to $\dot{S} = \int_a^b \sigma_t dz \geq 0 $ in accordance with the second law of Thermodynamics.
\end{lemma}
\begin{proof} The entropy production is encoded in the last line of the IPHS representation (entropy balance) (cf. \cite{RamirezCES2022}).
\end{proof}

\subsection{The 1D heat equation as an IPHS}
In order to illustrate the aforementioned representation we consider here the 1D heat equation example. The internal energy density $u(s)$ is chosen as thermodynamic potential function and $U(s)=\int_a^b u\, \d z$. From Gibb's equation one has $T=\frac{du}{ds}(s)$ leading to the following entropy balance equation
\begin{equation*}
\frac{\partial s}{\partial t} = -\frac{1}{T} \frac{\partial}{\partial z}\left( -\lambda\frac{\partial T}{\partial z} \right)\\
\end{equation*}
where, according to Fourier's law, $\lambda$ denotes the heat conduction coefficient and $-\lambda\frac{\partial T}{\partial z}=f_Q$ corresponds to the heat flux. \begin{equation}\label{heat_eq_ds}
\frac{\partial s}{\partial t} =\frac{\partial}{\partial z}\left(\frac{\lambda}{T}\frac{\partial T}{\partial z} \right) + \frac{\lambda}{T^2}\left(\frac{\partial T}{\partial z}\right)^2 \\
\end{equation}
from where the entropy production $\sigma_s =\frac{\lambda}{T^2}\left(\frac{\partial T}{\partial z}\right)^2$ is identified. This balance equation can be rewritten as:
\begin{equation*}
\frac{\partial s}{\partial t} = \frac{\lambda}{T^2}\frac{\partial T}{\partial z}\frac{\partial}{\partial z} \left(\frac{\delta U}{\delta s}\right) + \frac{\partial}{\partial z} \left(\frac{\lambda}{T^2}\frac{\partial T}{\partial z}\left(\frac{\delta U}{\delta s}\right)\right)
\end{equation*}
which is equivalent to \eqref{IPHS-BCS} where $P_0=0$, $P_1=0$, $G_0=0$, $G_1=0$, $g_s=1$ and $r_s=\gamma_s \{S | U \}$ with $\gamma_s=\frac{\lambda}{T^2}$ and $\{S | U \}=\frac{\partial T}{\partial z}$, i.e.
\[
\frac{\partial s}{\partial t}= \left( g_s \mathbf{r_s} \frac{\partial \left(.\right)}{\partial z}+\frac{\partial \left(g_s \mathbf{r_s}.\right)}{\partial z}\right) \frac{\delta U}{\delta s}\,.
\]
In this case $P_e=\frac{1}{2}\begin{bmatrix}0&1\\ 1& 0\end{bmatrix}$, $n=1$ and $m=1$.  Choosing $\Xi_1=\frac{1}{\sqrt{2}}\begin{bmatrix}
1 &  0 \\
1& 0 
\end{bmatrix}$, $\Xi_2=\frac{1}{\sqrt{2}}\begin{bmatrix}
0 &  1 \\
0& -1 
\end{bmatrix}$ the boundary inputs and outputs of the system are
\begin{align}%
\label{eq:input-output1D}
v(t)&=
\begin{bmatrix}
\left( \frac{\lambda_s}{T} \frac{\partial T}{\partial z}\right)(t,b) \\ 
-\left(\frac{\lambda_s}{T} \frac{\partial T}{\partial z}\right)(t,a)
\end{bmatrix}, & 
y(t)&=
\begin{bmatrix}
T(t,b)\\ 
T(t,a)
\end{bmatrix},
\end{align}
respectively the entropy flux and the temperature at each boundary.

In the following section, we will see how to extend this kind of formulation to multidimensional cases for diffusion processes.

\section{The $ND$ heat conduction as an IPHS}
\label{s-Cond}


Consider the heat conduction over an $N$D spatial domain. We assume the medium to be undeformable, i.e. its deformations are neglected, and consider only one physical domain, the thermal domain and its dynamics. The conserved quantity is the density of internal energy and the state reduces to a unique variable. Choose the internal energy density $u = u(s)$ as thermodynamic potential function and $U(s) = \int_\Dom u\,\d z$ , in this case Gibbs' relation defines the temperature as intensive variable conjugated to the extensive variable, the entropy by $T=\frac{\d u}{\d s}(s)$. This leads to write the following entropy balance equation
\begin{equation}
T\,\frac{\partial s}{\partial t} = \frac{\partial u}{\partial t} = -\div(\vf_Q)
= \div(\lambda\,\grad(T))
\end{equation}
where, according to Fourier's law, $\lambda$ denotes the heat conduction coefficient and $\vf_Q := - \lambda\,\grad(T)$ corresponds to the heat flux.
Alternatively, the heat conduction can also be written in terms of the entropy flux
$\vf_s:=\frac{1}{T}\,\vf_Q$
\begin{eqnarray}
\frac{\partial s}{\partial t} &=& -\frac{1}{T}\,\div(T\,\vf_s)
= -\div(\vf_s) -\frac{1}{T}\,\grad(T)\cdot\vf_s \\
&=&  \div(\frac{\lambda}{T}\,\grad(T)) +\frac{\lambda}{T^2}\,\|\grad(T)\|^2\,, \label{eq-jaum}
\label{eq:Jaumann}
\end{eqnarray}
from which the entropy production $\sigma_s := \frac{\lambda}{T^2}\,\|\grad(T)\|^2$ is readily identified. This form of the balance equation \eqref{eq-jaum} is also known as Jaumann's entropy balance.

Recalling that $\delta_s U = \partial_s u  = T$, the IPHS formulation of the heat conduction can be obtained from \eqref{eq:Jaumann} as follows: 
$$
\frac{\partial s}{\partial t} =
\frac{\lambda}{T^2}\,\grad(T)\cdot \grad(\delta_s U)
+ \div(\frac{\lambda}{T^2}\,\grad(T)\,\delta_s U)\,,
$$
which can be seen as a particular instance of the relation 
\begin{equation}
    \label{eq-heat}
\partial_t s = [g_s \vr_s\cdot\grad(.) + \div(g_s \vr_s .)] (\delta_s U)\,,
\end{equation}
with $g_s=1$,  $\vr_s=\gamma_s\, \left\{ S\,|\,  U \right\}$ with $\gamma_s = \frac{\lambda}{T^2}$ and $\left\{ S\,|\,  U \right\} = (\delta_s S)\, \grad(\delta_s U) = 1\, \grad(T) $. \\

\begin{proposition}\label{prop-Psi}
The unbounded differential operator  
\begin{equation} \label{eq-Psi}
    \Psi : T \mapsto \Psi(T):= [g_s \vr_s\cdot\grad(T) + \div(g_s \vr_s\, T)]
\end{equation}
is formally skew symmetric in $L^2(\Dom)$.
\end{proposition} 
\begin{proof}
$(\Psi(T), \theta)_{L^2} = \int_{\Dom} \theta\, [g_s \vr_s\cdot\grad(T) + \div(g_s \vr_s\, T)]\,\d z
= - \int_{\Dom}  T\, \div(g_s \vr_s\, \theta) + \grad(\theta) \cdot g_s \vr_s\,T \,\d z
=  - (T, \Psi(\theta))_{L^2}$ 
for fields $\theta$ and $T$ vanishing at the boundary $\Domb$. 
\end{proof}
Otherwise, 
$(\Psi(T), \theta)_{L^2} + (T, \Psi(\theta))_{L^2} = 2\,\int_\Domb g_s\,\vr_s\cdot \vn\,T\,\theta\,\d \gamma$
 in general.
One useful case to be inspected is the case with $\theta = T$, which gives 
\begin{equation}
    \label{eq-TPsiT}
\int_{\Dom} T\,\Psi(T)\,\d z = \int_\Domb T\,g_s\,\vr_s\cdot \vn\,T\,\d \gamma\,.
\end{equation}
  
Moreover, in \eqref{eq-heat}, we can identify the entropy production as
$$
\sigma_s = g_s \vr_s\cdot\grad(\delta_s U) = \gamma_s \, \|\left\{ S \,|\,  U \right\}\|^2 \geq 0\,.
$$


With  formulation \eqref{eq-heat} at hand, the first principle of Thermodynamics can be recovered in:

\begin{lemma}\label{lemma_3}(First law of Thermodynamics)
$$
\frac{\d}{\d t} U = \int_\Dom \partial_t u\,\d z =  - \int_\Domb \vf_s\cdot \vn\,(\delta_s U)\,\d \gamma\,,
$$
where the boundary term can be decomposed into  $u_\partial=-\vf_s\cdot\vn$, the incoming normal component of the entropy flux, and $y_\partial = T$, the temperature.
\end{lemma}
Notice that this is nothing but the $ND$ generalization of \eqref{eq:input-output1D}.

\begin{proof}
  Since $\partial_t u = -\div(\vf_Q)= - \div(T\,\vf_s)$, making use of Stokes theorem once again, we see that $\frac{\d}{\d t} U = - \int_\Domb T\,\vf_s\cdot\vn\,\d \gamma$.
  One can indeed choose  $u_\partial=-\vf_s\cdot\vn$ and $y_\partial = T$ as collocated boundary controls for the internal energy.
\end{proof}

From the last equation of the IPHS formulation, the second principle of Thermodynamics can be recovered in:

\begin{lemma}[Second law of Thermodynamics]\label{lemma_4} 
$$
\frac{\d}{\d t} U = \int_\Dom \partial_t s\,\d z = \int_\Dom \sigma_s\,\d z + \int_\Domb g_s \vr_s\cdot \vn\,(\delta_s U)\,\d \gamma\,,
$$
where the volumic term is the entropy creation $\sigma_s \geq 0$, and the boundary term can be decomposed into  $u_{s\partial}=-\vf_Q\cdot\vn$, the incoming normal component of the heat flux, and $y_{s \partial} = \beta:=\frac{1}{T}$, the reciprocal temperature.
\end{lemma} 
\begin{proof}
  The internal entropy production is clearly identified as $\int_\Dom \sigma_s\,\d z \geq 0$. The boundary term can be computed as
  \begin{eqnarray*}
      \int_\Dom \div(g_s \vr_s\,\delta_s U)\,\d z &=& \int_\Domb g_s \vr_s\cdot\vn\,(\delta_s U)\,\d \gamma \\
  &=&  \int_\Domb \frac{1}{T}\,(\lambda\,\vn\cdot\grad(T))\,\d \gamma\,.
  \end{eqnarray*}
Thus, as collocated boundary controls for the entropy, one can indeed choose  $u_{s\partial}=-\vf_Q\cdot\vn= \lambda\,\vn\cdot\grad(T)$ and $y_{s \partial} = \beta:=\frac{1}{T}$.
\end{proof}

\section{The complete $ND$ diffusion process as an IPHS}
\label{s-CDP}

We consider now $ND$ diffusion processes involving one or more chemical species in addition to thermal conduction. 

\subsection{Processes with diffusion of one species only}
In the case of diffusion of one species only in $ND$ the balance equations on the extensive variables, namely the molar concentration $c$ and the entropy density $s$ in this case, are given by
\begin{eqnarray}
\frac{\partial c}{\partial t} &=& - \div(\vf_c) \label{eq-1n}\\
\frac{\partial s}{\partial t} &=& - \div(\vf_s) + \sigma_s  + \sigma_c \label{eq-1s}
\end{eqnarray}
where $\vf_c := -d\,\grad(\mu)$ is the flux of moles, with $d>0$ the diffusion coefficient from Fick's law and $\mu$ the chemical potentiel function;  $\vf_s := -\frac{\lambda}{T} \grad(T)$ is the entropy flux. The internal entropy production $\sigma_c$ due to irreversible diffusion is  $\sigma_c := -\frac{1}{T}\vf_c \cdot \grad(\mu) = \frac{d}{T} \| \grad(\mu) \|^2 \geq 0$, and  the internal entropy production $\sigma_s$  due to heat conduction, $\sigma_s := -\frac{1}{T}\vf_s \cdot \grad(T) = \frac{\lambda}{T^2} \| \grad(T) \|^2 \geq 0$.
We define as state variables $\vx := \begin{bmatrix}
c & s
\end{bmatrix}^\top$ 
and co-state or co-energy variables w.r.t. $U(\vx)$, the internal energy: $\ve:=\delta_{\vx} U :=  \begin{bmatrix}
\frac{\partial u}{\partial c} & \frac{\partial u}{\partial s}
\end{bmatrix}^\top
= \begin{bmatrix}
\mu & T
\end{bmatrix}^\top$. 
In this case we have 
$\delta_{\vx} S :=  \begin{bmatrix}
\frac{\partial s}{\partial c} & \frac{\partial s}{\partial s}
\end{bmatrix}^\top
= \begin{bmatrix}
0 & 1
\end{bmatrix}^\top$. We also define 
\begin{itemize}
\item  $\vr_s:=\gamma_s\, \left\{ S \,|\,  U \right\}$, the modulating function of the heat conduction, with $\gamma_s = \frac{\lambda}{T^2}$ and $\left\{ S\,|\,  U \right\} = (\delta_s S)\, \grad(\delta_s U) = 1\, \grad(T)$ which is the driving force of temperature.
\item $\vr_c:=\gamma_c\, \left\{ S \,| 1 \, |\,  U \right\}$, the modulating function of the mass diffusion of the species, with $\gamma_c = \frac{d}{T}$  and $\left\{ S \,| 1 \, |\, U \right\} = (\delta_s S)\, \grad(\delta_c U) = 1\, \grad(\mu)$, which is the driving force of the diffusion.
\end{itemize}



With the newly defined quantities, system \eqref{eq-1n}-\eqref{eq-1s} can be rewritten as follows:
\begin{align}
 \partial_t c &= - \div(\vf_c)  = +\div(\vr_c\, \frac{\partial u}{\partial s})\,,\\
 \partial_t s &= - \div(\vf_s) + \sigma_s + \sigma_c\,,  \nonumber \\
        &= +\div(\vr_s\, \frac{\partial u}{\partial s})   +  \vr_s\cdot \grad(\frac{\partial u}{\partial s}) + \vr_c\cdot \grad(\frac{\partial u}{\partial c})\,, \nonumber \\
  &= \vr_c\cdot \grad(\frac{\partial u}{\partial c}) + \Psi(\frac{\partial u}{\partial s}) \,.
\end{align}
 
In more compact notations, the following structure has now been obtained for the system of equations \eqref{eq-1n}-\eqref{eq-1s}:
\begin{equation}
\label{eq-CDScompact}
    \begin{bmatrix}
\frac{\partial c}{\partial t} \\ 
\frac{\partial s}{\partial t}
\end{bmatrix}
=
\begin{bmatrix}
0 & \div(\vr_c \cdot) \\
\vr_c \cdot \grad(\cdot) & [g_s \vr_s\cdot\grad(\cdot) + \div(g_s \vr_s \cdot)]
\end{bmatrix}
\,
\begin{bmatrix}
\frac{\partial u}{\partial c} \\ 
\frac{\partial u}{\partial s}
\end{bmatrix}\,.
\end{equation}


We have already identified in Proposition~\ref{prop-Psi} that the unbounded operator  $\Psi$ in \eqref{eq-Psi} is formally skew symmetric. Moreover, it is also a classical computation to check that the adjoint of $T \mapsto \div(\vr_c T)$ is $\mu \mapsto  - \vr_c \cdot \grad(\mu)$; hence the {\text glob}al operator $\cal J$ in \eqref{eq-CDScompact} is formally skew symmetric.


With formulation \eqref{eq-CDScompact} at hand, the second principle of Thermodynamics can be recovered in: \\

\begin{lemma} [Second law of Thermodynamics] 
\label{lm-2Th-1species}
$$
\frac{\d}{\d t} {\cal S} = \int_\Dom \partial_t s\,\d z
= \int_\Dom (\sigma_s + \sigma_c)\,\d z + \int_\Domb g_s \vr_s\cdot \vn\,(\delta_s U)\,\d \gamma\,,
$$
where the volumic term is the total entropy creation $\sigma_s + \sigma_c \geq 0$, and the boundary term can be decomposed into  $u_{s\partial}=-\vf_Q\cdot\vn$, the incoming normal component of the heat flux, and $y_{s \partial} = \beta:=\frac{1}{T}$, the reciprocal temperature.
\end{lemma}
\begin{proof} From the last line, 
    $\int_\Dom \partial_t s\,\d z =  \int_\Domb g_s \vr_s \cdot \vn\, T\, \d \gamma\, + \, 
    \int_\Dom  [\vr_c \cdot \grad(\frac{\partial u}{\partial c}) + g_s \vr_s\cdot\grad(\frac{\partial u}{\partial s})]\, \d z$\,,
    and upon substitution, $\vr_c \cdot \grad(\mu) = \frac{d}{T} \|\grad(\mu)\|^2:=\sigma_c$ first, and $g_s \vr_s\cdot\grad(T)= \frac{\lambda}{T^2} \|\grad(T\|^2 := \sigma_s$ second.
\end{proof}

With  formulation  \eqref{eq-CDScompact} at hand, the first principle of Thermodynamics can be recovered in: \\

\begin{lemma} [First law of Thermodynamics] 
\label{lm-1Th-1species}
$$
\frac{\d}{\d t} U = \int_\Dom \partial_t u\,\d z
=  \int_\Domb T\, \frac{\lambda}{T} \vn\cdot\grad(T) +  \mu\, d\vn\cdot\grad(\mu)\,\d \gamma\,,
$$
where the boundary term can be decomposed into two parts,  $u_\partial^s=-\vf_s\cdot\vn$   and $u_\partial^n=-\vf_c\cdot\vn$, and $y_\partial^s = T$ and $y_\partial^n = \mu$, for instance.
\end{lemma}
\begin{proof} It is based on the chain rule $\partial_t u = \frac{\partial u}{\partial c}\,\partial_t c +\frac{\partial u}{\partial s} \,\partial_t s $. 
     \begin{eqnarray*}
     \int_\Dom \partial_t u\,\d z &=& \int_\Dom  \begin{bmatrix}
\frac{\partial u}{\partial c}  \frac{\partial u}{\partial s}
\end{bmatrix}\,
{\cal J}\, \begin{bmatrix}
\frac{\partial u}{\partial c} \\ 
\frac{\partial u}{\partial s}
\end{bmatrix} \,\d z \\
&=& \int_\Dom \div(\mu \vr_c T) +  T\,\Psi(T) \,\d z\,,\\
&=& \int_\Domb [\mu\,\vn\cdot\vr_c\,T +  T\, g_s\vn\cdot\vr_s\,T]\,\d \gamma\,, \text{thanks to } \eqref{eq-TPsiT}\,,\\
&=& \int_\Domb [\mu\, (d\vn\cdot\grad(\mu)) + T \,(\frac{\lambda}{T}\,\vn\cdot\grad(T))]\,\d \gamma\,.
\end{eqnarray*}
Indeed, since ${\cal J}$ in \eqref{eq-CDScompact} is formally skew symmetric, only boundary terms appear in the result. 
\end{proof}


\subsection{Processes with diffusion of $n$ different species}
In a similar way as in the diffusion of one species, we write the balance equations on the extensive variables of Thermodynamics, i.e. the $n$ molar concentrations of the species involved in the diffusion process plus the entropy density, resulting in the $n+1$ following equations:
\begin{align}
\frac{\partial c_1}{\partial t} &= - \div(\vf_{c_1})\,, \label{eq-n1}\\
  & \vdots   \nonumber \\ 
\frac{\partial c_n}{\partial t} &= - \div(\vf_{c_n})\,, \label{eq-nM} \\
\frac{\partial s}{\partial t} &= - \div(\vf_s) + \sigma_s  + \sum_{i=1}^n \sigma_{c_i}\,, \label{eq-ns}
\end{align}
where thanks to Fick's law, $\vf_{c_i}:= -d_i\,\grad(\mu_i)$,  and 
$\sigma_{c_i} :=   -\frac{1}{T}\vf_{c_i} \cdot \grad(\mu_i) = \frac{d_i}{T} \| \grad(\mu_i) \|^2 \geq 0$ for the entropy production.

Let us now define $\vr_{c_i}:=-\frac{1}{T}\vf_{c_i}= +\frac{d_i}{T}\,\grad(\mu_i)$, the modulating function of the mass diffusion of the $i$-th species, then we can compute $- \div(\vf_{c_i}) = + \div(\vr_{c_i}\,\frac{\partial u}{\partial s})$, and $\sigma_{c_i}= \vr_{c_i} \cdot \grad(\frac{\partial u}{\partial {c_i}})$. Hence, in more compact notations, the following structure has been obtained for \eqref{eq-n1}-\eqref{eq-ns}:
\begin{equation}
\label{eq-nCDScompact0}
    \begin{bmatrix}
\frac{\partial c_1}{\partial t} \\ 
\vdots \\ 
\frac{\partial c_n}{\partial t}\\ 
\frac{\partial s}{\partial t}
\end{bmatrix}
= 
{\cal J}_{\text{glob}} \,
\begin{bmatrix}
\frac{\partial u}{\partial c_1} \\ 
\vdots \\  
\frac{\partial u}{\partial c_n} \\
\frac{\partial u}{\partial s}
\end{bmatrix}\,,
\end{equation}
where
\begin{equation}
\label{eq-nCDScompact}
{\cal J}_{\text{glob}} 
:=
\begin{bmatrix}
0 & \cdots & 0 & \div(\vr_{c_1} \cdot) \\
\vdots & & \vdots & \vdots \\
0 & \cdots & 0 & \div(\vr_{c_n} \cdot) \\
[\vr_{c_1}\cdot \grad(\cdot)]  & \cdots & [\vr_{c_n}\cdot \grad(\cdot)]  & \Psi(\cdot)
\end{bmatrix}.
\end{equation}

We have already identified in Proposition~\ref{prop-Psi} that the unbounded operator  $\Psi$ in \eqref{eq-Psi} is formally skew symmetric. Moreover, it is also a classical computation to check that the adjoint of $T \mapsto \div(\vr_{c_i} T)$ is $\mu \mapsto  - \vr_{c_i} \cdot \grad(\mu)$; hence the global operator ${\cal J}_{\text {glob}}$ in \eqref{eq-nCDScompact} is formally skew symmetric.


With formulation \eqref{eq-nCDScompact} at hand, the second principle of Thermodynamics can be recovered in: 

\begin{lemma}[Second law of Thermodynamics]\label{lemma-7} 
$$
\frac{\d}{\d t} S = \int_\Dom \partial_t s\,\d z
= \int_\Dom (\sigma_s + \sum_{i=1}^n \sigma_{c_i})\,\d z + \int_\Domb g_s \vr_s\cdot \vn\,(\delta_s U )\,\d \gamma\,,
$$
where the volumic term is the total entropy creation $\sigma_s + \sum_{i=1}^n \sigma_{c_i} \geq 0$, and the boundary term can be decomposed into  $u_{s\partial}=-\vf_Q\cdot\vn$, the incoming normal component of the heat flux, and $y_{s \partial} = \beta:=\frac{1}{T}$, the reciprocal temperature.
\end{lemma}
\begin{proof}
The proof is very similar to that of Lemma~\ref{lm-2Th-1species}, summing on $i$ between $1$ and $n$.
\end{proof}

With formulation \eqref{eq-nCDScompact}  at hand, the first principle of Thermodynamics can be recovered in:

\begin{lemma}[First law of Thermodynamics]\label{lemma-8} 
\begin{eqnarray*}
\frac{\d}{\d t} U &=& \int_\Dom \partial_t u\,\d z\,, \\
&=&  \int_\Domb T\, (\frac{\lambda}{T} \vn\cdot\grad(T)) +  \sum_{i=1}^n \mu_i\, (d_i\vn\cdot\grad(\mu_i))\,\d \gamma\,,
\end{eqnarray*}
where the boundary term can be decomposed into $1+n$ terms,  $u_\partial^s=-\vf_s\cdot\vn$ and $u_\partial^{c_i}=-\vf_{c_i}\cdot\vn$  for $1\leq i \leq n$, and $y_\partial^s = T$ and $y_\partial^{c_i} = \mu_i$ for $1\leq i \leq n$, for instance.
\end{lemma}
\begin{proof}
The proof is very similar to that of Lemma~\ref{lm-1Th-1species}, summing on $i$ between $1$ and $n$.
\end{proof}

Let us enlighten the structure which has been obtained above through a factorization result that happens to be very meaningful from the thermodynamical point of view.

\begin{proposition}\label{prop-n-diffusion}
The operator ${\cal J}_{\text{glob}}$ in \eqref{eq-nCDScompact} can be factorized in the following way. Let 
$$
R_1 : \vartheta \mapsto 
\begin{bmatrix}
\vr_{c_1}\\ 
\vdots \\ 
\vr_{c_n}
\end{bmatrix} \, \vartheta\,,
$$
and let 
\begin{equation} \label{eq-G1}
{\cal G}_1 : 
\begin{bmatrix}
\vr_{1} \\ 
\vdots \\ 
\vr_{n}
\end{bmatrix} \mapsto
\begin{bmatrix}
\div(\vr_{1})\\ 
\vdots \\ 
\div(\vr_{n})
\end{bmatrix}\,,
\end{equation}
then 
$$ 
{\cal J}_{\text{glob}} = 
\begin{bmatrix}
0 && {\cal G}_1\,R_1 \\
- R_1^*\,{\cal G}_1^* && \Psi(\cdot)
\end{bmatrix}\,.
$$
\end{proposition}
\begin{proof}
The identification of the formal adjoints gives the structure. Indeed, 
$$ 
{\cal G}_1^* :
\begin{bmatrix}
d_{1}\\ 
\vdots \\ 
d_{n}
\end{bmatrix} \mapsto
\begin{bmatrix}
 - \grad(d_1) \\ 
 \vdots \\ 
 - \grad(d_n) 
\end{bmatrix}\,,
 $$
 and 
 $$R_1^* : 
 \begin{bmatrix}
\vg_1 \\ 
\vdots \\ 
\vg_n 
\end{bmatrix}
\mapsto 
\vr_{c_1} \cdot \vg_1 + \cdots + \vr_{c_n} \cdot \vg_n
\,.$$
\end{proof}

\subsection{The general $ND$ formulation}

At this stage we are able to propose the following structured PDE for IPHS in $ND$ (compare with \eqref{IPHS-BCS})

\begin{proposition}
\begin{multline}\label{IPHS-BCSnD}
\frac{\partial }{\partial t} 
\begin{bmatrix} 
x(t,z)\\ 
s(t,z)
\end{bmatrix}=
\begin{bmatrix}
P_0 & G_0\mathbf{R_0}\\ 
-\mathbf{R_0}^\top G_0^\top &0
\end{bmatrix}
\begin{bmatrix}
\frac{\delta H}{\delta x}(t,z)\\ 
\frac{\delta H}{\delta s}(t,z) 
\end{bmatrix} + \\ 
\begin{bmatrix} {\cal J}_1  && {\cal G}_1 \mathbf{R_1}  \\ 
- \mathbf{R_1}^* {\cal G}_1^*   && g_s \mathbf{r_s}\cdot \grad\left(\cdot\right) + \div \left(g_s \mathbf{r_s} \cdot\right)\end{bmatrix} 
\begin{bmatrix}\frac{\delta H}{\delta x}(t,z)\\ \frac{\delta H}{\delta s}(t,z) \end{bmatrix}
\end{multline}
where $P_0=-P_0^\top$, ${\cal J}_1 = - {\cal J}_1^*$ and ${\cal G}_1 $ is a first-order differential operator. 
\end{proposition}

\begin{proof}
The proof follows from Propositions \ref{prop-Psi} and \ref{prop-n-diffusion} and Lemmas \ref{lemma_3},\ref{lemma_4}, \ref{lemma-7} and \ref{lemma-8}.     
\end{proof}

Notice that all the examples presented in the present  paper do fit in this general framework (e.g. ${\cal G}_1 :=G_1 \frac{\partial}{\partial z}$ in \eqref{IPHS-BCS}, or \eqref{eq-G1}), as well as the case of the {\em non-isentropic fluid}, treated in depth in \citep{MoraWC2023}.
Moreover, following \citep[\S.5]{RamirezCES2022}, the example of  reaction-diffusion systems would only add $P_0$ and $G_0$ terms in the dynamics.

\section{Conclusion and Future work} \label{s-CF}

The 1D IPHS formulations of \citep{RamirezCES2022} has been extended to boundary-controlled $ND$ distributed parameter systems describing conduction–diffusion fluid phenomena. Within a unified and thermodynamically consistent framework, it has been showed that conduction, diffusion, and reactive transport can be integrated into a single coherent structure that preserves both global energy balance and the proper representation of entropy production. This extension provides a basis for systematic modeling and control of complex multi-physical processes governed by coupled transport mechanisms. \\
As future work, the inclusion of anisotropy should be done both for conduction and diffusions: in this case, the strictly positive scalar parameters $\lambda$ in Fourier's law and $d_i$ in Fick's law will become positive-definite symmetric tensors. Then, the example of a reaction-diffusion process in $ND$ should be carried out, based on the 1-D example. The interest of such  IPHS structured formulations for numerical simulations should be seen as a longer term perspective, see e.g. \citep{CardosoJCF2025}: indeed, structure-preserving numerical methods,  such as the Partitioned Finite Element Method (PFEM) introduced in \citep{CardosoIMAJMCI2020},  should be able to guarantee both the Thermodynamics principles at the discrete level.

{\normalsize
\bibliography{sample}
}
\end{document}